\documentclass[12pt]{article}

\usepackage[hmargin=1.3in,vmargin=1.3in]{geometry}
\usepackage{bbding}
\usepackage{mathrsfs}
\usepackage{cite}
\usepackage{amsfonts}
\usepackage{multirow}
\usepackage{amsfonts,amssymb,amsmath,amsthm,bm}
\usepackage[colorlinks,
            linkcolor=blue,
            anchorcolor=blue,
            citecolor=blue]{hyperref}

\newtheorem{theorem}{Theorem}
\newtheorem{example}{Example}

\newtheorem{lemma}{Lemma}

\begin{document}

\title{Construction of MDS Euclidean Self-Dual Codes via Two Subsets}
\author{\small Weijun Fang$^{1,2,}$\thanks{Corresponding Author}  \ Shu-Tao Xia$^{1,2}$ \ Fang-Wei Fu$^{3,4}$ \\
\scriptsize $^1$ Tsinghua Shenzhen International Graduate School, Tsinghua University, Shenzhen 518055,  China \\
\scriptsize $^2$ PCL Research Center of Networks and Communications, Peng Cheng Laboratory, Shenzhen 518055,  China\\
\scriptsize $^3$  Chern Institute of Mathematics and LPMC, Nankai University, Tianjin 300071,  China\\
\scriptsize $^4$ Tianjin Key Laboratory of Network and Data Security Technology, Nankai University, Tianjin 300071,  China\\
\scriptsize E-mail: nankaifwj@163.com, xiast@sz.tsinghua.edu.cn, fwfu@nankai.edu.cn\\
}
\date{}
\maketitle
\thispagestyle{empty}
\begin{abstract}
The parameters of a $q$-ary MDS Euclidean self-dual codes are completely determined by its length and the construction of MDS Euclidean self-dual codes with new length has been widely investigated in recent years. In this paper, we give a further study on the construction of MDS Euclidean self-dual codes via generalized Reed-Solomon (GRS) codes and their extended codes.  The main idea of our construction is to choose suitable evaluation points such that the corresponding (extended) GRS codes are Euclidean self-dual.

Firstly, we consider the evaluation set consists of two disjoint subsets, one of which is based on the trace function, the other one is a union of a subspace and its cosets. Then four new families of MDS Euclidean self-dual codes are constructed. Secondly, we give a simple but useful lemma to ensure that the symmetric difference of two intersecting subsets of finite fields can be taken as the desired evaluation set. Based on this lemma, we generalize our first construction and provide two new families of MDS Euclidean self-dual codes. Finally, by using two multiplicative subgroups and their cosets which have nonempty intersection, we present three generic constructions of  MDS Euclidean self-dual codes with flexible parameters. Several new families of MDS Euclidean self-dual codes are explicitly constructed.

\end{abstract}


\maketitle

\section{Introduction}

Due to their nice properties and wide applications, MDS codes and Euclidean self-dual codes are two important classes of linear codes in coding theory. Let $q$ be a prime power and $\mathbb{F}_{q}$ be the finite field with $q$ elements. A $q$-ary $[n, k, d]$-linear code is defined as a $k$-dimensional subspace of $\mathbb{F}^{n}_{q}$ with minimum Hamming distance $d$. One of the relations among these parameters is the well-known Singleton bound, which says that any $[n, k, d]$-linear code has to satisfy that
\[d \leq n-k+1.\]
An $[n, k, d]$-linear code $C$ is called a \emph{maximum distance separable} (MDS) code if it achieves the Singleton bound with equality. MDS codes are closely related to some other mathematical aspects, such as the orthogonal arrays in combinatorial design and the $n$-arcs in finite geometry (see \cite[Chap. 11]{MS77}). MDS codes also have widespread applications in data storage, such as coding for distributed storage systems. As an most important class of MDS codes, Reed-Solomon codes have important applications in engineering due to their easy encoding and efficient decoding algorithm.

For any two vectors $\textbf{x}= (x_{1}, \ldots, x_{n}) \in \mathbb{F}_{q}^{n}$ and $\textbf{y}= (y_{1}, \ldots, y_{n}) \in \mathbb{F}_{q}^{n}$, we define their Euclidean inner product as
\[\langle \textbf{x}, \textbf{y}\rangle = \sum_{i=1}^{n} x_i y_i.\]
The Euclidean dual code $C^{\perp}$ of $C$ then is given as
\[C^{\perp} := \{\textbf{x} \in \mathbb{F}_{q}^{n}: \langle \textbf{x}, \textbf{y}\rangle=0, \textnormal{ for any } \textbf{y} \in C \}.\]
$C$ is called a \emph{Euclidean self-dual} code if $C^{\perp}=C$. It is obvious that the length of a Euclidean self-dual code is even. It has been proved in \cite{P68} that a $q$-ary Euclidean self-dual code of even length $n$ exists if and only if $(-1)^{\frac{n}{2}}$ is a square element in $\mathbb{F}_{q}$. There are some well-known linear codes which are also Euclidean self-dual codes, such as the $[24, 12, 8]_{2}$-extended binary Golay code, the $[12, 6, 6]_{3}$-extended ternary Golay code and the Pless symmetry codes (see \cite{HP03}). Euclidean self-dual codes have also been found various interesting applications in other aspects. In \cite{C08} and \cite{DMS08}, the authors established the connections between self-dual codes and linear secret sharing schemes (LSSSs). Euclidean self-dual codes are closely related to combinatorics and unimodular integer lattices (see \cite{HP03, CS99}). Therefore, it is of great interest to investigate the MDS Euclidean self-dual codes. One of the basic problem for this theme is to determine existence of MDS Euclidean self-dual codes. Note that a $q$-ary MDS Euclidean self-dual code of length $n$ has dimension $\frac{n}{2}$ and minimum distance $\frac{n}{2}+1$. So it is sufficient to consider the problem for which lengths an MDS Euclidean self-dual code over $\mathbb{F}_{q}$ exists.

\subsection*{Related Work}
In recent years, this problem has been extensively studied \cite{GK02,BGGHK03,BBDW04,KLISIT04,KL04,HL06,GKL08,YC15}. For $q$ is even, Grassl and Gulliver \cite{GG08} proved that there is a $q$-ary MDS Euclidean self-dual code of even length $n$ for all $n \leq q$. Some new MDS Euclidean self-dual codes were obtained through cyclic and constacyclic codes in \cite{G12} and \cite{TW17}. In \cite{JX17}, Jin and Xing first presented a systematic approach to construct MDS Euclidean self-dual codes by utilizing generalized Reed-Solomon (GRS) codes over finite fields. Several new classes of MDS Euclidean self-dual codes are obtained by choosing suitable evaluation points. Since then, GRS codes becomes one class of the most commonly used tools to construct MDS Euclidean self-dual codes. Yan \cite{Y18}, Fang and Fu \cite{FF19} generalized this method to extended GRS codes and construct several new families of  MDS Euclidean self-dual codes with flexible parameters. Zhang and Feng \cite{ZF19} presented a unified approach to obtain some known results with concise statements and simplified proofs. In \cite{FLLL19,LLL19}, the authors constructed some families of Euclidean self-dual GRS codes with new parameters via multiplicative subgroups. As far as we known,  most of previously known results obtained from GRS codes considered the evaluation set consists of a multiplicative subgroup of $\mathbb{F}_{q}^{*}$ and its cosets, or a subspace of $\mathbb{F}_{q}$ and its cosets. In \cite{FLL19}, Fang \emph{et al.} took the evaluation set as the union of two disjoint multiplicative subgroups and their cosets.  In \cite{ZF19-2}, Zhang and Feng presented some new constructions of MDS Euclidean self-dual codes via cyclotomy. In \cite{S19}, Sok gave some explicit constructions of MDS Euclidean self-dual codes via rational function fields. In the following Table I, we summary some known results about the construction of MDS Euclidean self-dual codes.

From Table I, significant progress has been made on the construction of MDS Euclidean self-dual codes. However, it is still a great challenge to determine the existence of  $q$-ary MDS Euclidean self-dual codes of length $n$ for all possible even $n \leq q+1$.

\newcommand{\tabincell}[2]{\begin{tabular}{@{}#1@{}}#2\end{tabular}}
\begin{table}[htbp]
\scriptsize
\centering
\caption{Some known results on MDS Euclidean self-dual codes of even length $n$}
\vskip 3mm
\begin{tabular}{|c|c|c|}
  \hline
  $q$ & $n$ & References \\
  \hline
  $q$ even  & $n \leq q$ & \cite{GG08} \\
  \hline
  $q$ odd  & $n = q+1$ & \cite{GG08,JX17} \\
  \hline
  $q=r^{2}$ & $n \leq r$ & \cite{JX17} \\
  \hline
  $q=r^{2}$, $r \equiv 3 (\textnormal{mod } 4)$ & $n=2tr, t \leq \frac{r-1}{2}$ & \cite{JX17} \\
  \hline
  $q \equiv 1 (\textnormal{mod } 4)$ & $4^{n}n^{2} \leq q$ & \cite{JX17} \\
  \hline
  $q \equiv \textnormal{3 (mod 4)}$ & $n \equiv \textnormal{0 (mod 4)}$ and $(n-1) \mid (q-1)$ & \cite{TW17}\\
  \hline
  $q \equiv 1 (\textnormal{mod } 4)$ & $(n-1) \mid (q-1)$ & \cite{TW17} \\
  \hline
  $q \equiv 1 (\textnormal{mod } 4)$ & $n=2p^{\ell}, \ell \leq m$ & \cite{FF19} \\
  \hline
  $q \equiv 1 (\textnormal{mod } 4)$ & $n=p^{\ell}+1, \ell \leq m$ & \cite{FF19} \\
  \hline
  $q=r^{s}$, $r$ odd, $s$ even & \tabincell{c}{$n=2tr^{\ell}$, $0 \leq \ell <s$, \\ and $1 \leq t \leq \frac{r-1}{2}$} & \cite{FF19} \\
  \hline
  $q=r^{s}$, $r$ odd, $s$ is even & \tabincell{c}{$n=(2t+1)r^{\ell}+1$, $0 \leq \ell <s$, \\and $0 \leq t \leq \frac{r-1}{2}$} & \cite{FF19} \\
  \hline
  $q$ odd & $(n-2) \mid (q-1)$, $\eta(2-n)=1$ & \cite{FF19,Y18} \\
  \hline
  $q \equiv 1 (\textnormal{mod } 4)$ & $n \mid (q-1)$ & \cite{Y18} \\
  \hline
  $q=r^{s}$, $r$ odd, $s \geq 2$ & $n=tr$, $t$ even  and  $2t \mid (r-1)$ &\cite{Y18}\\
  \hline
  $q=r^{s}$, $r$ odd, $s \geq 2$ & \tabincell{c}{$n=tr$, $t$ even, $(t-1) \mid (r-1)$ \\and $\eta(1-t)=1$} &\cite{Y18}\\
  \hline
  $q=r^{s}$, $r$ odd, $s \geq 2$ & \tabincell{c}{$n=tr+1$, $t$ odd, $t \mid (r-1)$ \\and $\eta(t)=1$} &\cite{Y18}\\
  \hline
  $q=r^{s}$, $r$ odd, $s \geq 2$ & \tabincell{c}{$n=tr+1$, $t$ odd, $(t-1) \mid (r-1)$\\ and $\eta(t-1)=\eta(-1)=1$} &\cite{Y18}\\
  \hline
  $q=r^{2}$, $r$ odd & \tabincell{c}{$n=tm$, $\frac{q-1}{m}$ even, \\ and $1 \leq t \leq \frac{r+1}{\gcd(r+1, m)}$} & \cite{FLLL19}\\
  \hline
  $q=r^{2}$, $r$ odd & \tabincell{c}{$n=tm+1$, $tm$ odd, $m \mid (q-1)$, \\ and $2 \leq t \leq \frac{r+1}{2\gcd(r+1, m)}$} & \cite{FLLL19}\\
  \hline
  $q=r^{2}$, $r$ odd & \tabincell{c}{$n=tm+2$, $m \mid (q-1)$, $tm$ even \\(except when $t, m$ are even and $r \equiv 1 (\textnormal{mod } 4)$),  \\ and $1 \leq t \leq \frac{r+1}{\gcd(r+1, m)}$} & \cite{FLLL19}\\
  \hline
  $q=r^{2}$, $r$ odd & \tabincell{c}{$n=tm$, $\frac{q-1}{m}$ even, $1 \leq t \leq \frac{s(r-1)}{\gcd(s(r-1), m)}$\\ $s$ even, $s \mid m$, and $\frac{r+1}{s}$ even } & \cite{FLLL19}\\
  \hline
  $q=r^{2}$, $r$ odd & \tabincell{c}{$n=tm+2$, $\frac{q-1}{m}$ even, $1 \leq t \leq \frac{s(r-1)}{\gcd(s(r-1), m)}$\\ $s$ even, $s \mid m$, and $\frac{r+1}{s}$ even } & \cite{FLLL19}\\
  \hline
  $q=r^{2}$, $r$ odd & \tabincell{c}{$n=tm$, $\frac{q-1}{m}$ even, \\ and $1 \leq t \leq \frac{r-1}{\gcd(r-1, m)}$} & \cite{LLL19}\\
  \hline
  $q=r^{2}$, $r$ odd & \tabincell{c}{$n=tm+1$, $tm$ odd, $m \mid (q-1)$, \\ and $2 \leq t \leq \frac{r-1}{\gcd(r-1, m)}$} & \cite{LLL19}\\
  \hline
  $q=r^{2}$, $r$ odd & \tabincell{c}{$n=tm+2$, $tm$ even, $m \mid (q-1)$, \\ and $2 \leq t \leq \frac{r-1}{\gcd(r-1, m)}$} & \cite{LLL19}\\
  \hline
  $q=r^{2}$, $r \equiv 1 (\textnormal{mod } 4)$ & \tabincell{c}{$n=s(r-1)+t(r+1)$, $s$ even, \\$1 \leq s \leq \frac{r+1}{2}$ $1 \leq t \leq \frac{r-1}{2}$} & \cite{FLL19} \\
  \hline
  $q=r^{2}$, $r \equiv 3 (\textnormal{mod } 4)$ & \tabincell{c}{$n=s(r-1)+t(r+1)$, $s$ odd, \\$1 \leq s \leq \frac{r+1}{2}$ $1 \leq t \leq \frac{r-1}{2}$} & \cite{FLL19} \\
  \hline
\end{tabular}
\end{table}

\subsection*{Main Results}
In this paper, we further study the construction of MDS Euclidean self-dual codes with new parameters by using (extended) GRS codes over finite fields. The key point of our constructions is to choose suitable evaluation set $A \subseteq \mathbb{F}_{q}$ such that $\eta(\delta_{A}(a))$ are the same for all $a \in A$ or $\eta(-\delta_{A}(a))=1$ for all $a \in A$ (see Lemmas \ref{lem2} and \ref{lem3}), where $\eta(x)$ is the quadratic character of $\mathbb{F}_{q}^{*}$. The evaluation sets in our constructions are based on two subsets of finite fields.

More precisely, we first consider the set of evaluation points consists of two subsets, one of which is constructed from the trace function over finite fields, and the other one is the union of a suitable subspace of $\mathbb{F}_{q}$ and its cosets. Our first new construction of MDS Euclidean self-dual codes (see Theorems \ref{thm1} and \ref{thm2}) is then given when these two subsets are disjoint. Secondly, we provide a simple but useful lemma (see Lemma \ref{lem5}), which provides a sufficient condition such that the symmetric difference of two intersecting subsets satisfies the conditions in Lemmas \ref{lem2} or \ref{lem3}. Based on Lemma \ref{lem5}, we generalize our first construction and obtain two new families of MDS Euclidean self-dual codes (see Theorems \ref{3} and \ref{4}). Finally, by using two multiplicative subgroups and their cosets and Lemma \ref{lem5}, we present three generic constructions of MDS Euclidean self-dual codes with flexible parameters (see Theorems \ref{thm5}-\ref{thm7}). From these three powerful constructions, we can obtain several new families of MDS Euclidean self-dual codes by choosing different pairs $(\mu, \nu)$ (see Theorems \ref{thm8}-\ref{thm16}). To the best of our knowledge, this is the first systematic construction of MDS Euclidean self-dual codes by choosing two intersecting subsets as the evaluation set.

Let $p$ be an odd prime, $q=r^{2}$ and $r=p^{m}$ for some positive integer $m$. We summary our main results as follows.  If one of the following conditions holds, then there exists an MDS Euclidean self-dual codes of length $N$ over $\mathbb{F}_{q}$.

\begin{description}
  \item[(1)] $N=tr+sp^{\lceil\log_{p}(t)\rceil}$, for any even $t$ and even $s$ with $1 \leq t \leq r$ and $0 \leq s \leq p^{m-\lceil\log_{p}(t)\rceil}-1$; (see Theorem \ref{thm1})
  \item[(2)] $N=tr+sp^{\lceil\log_{p}(t)\rceil}+1$, for any odd $t$ and even $s$ with $1 \leq t \leq r$ and $0 \leq s \leq p^{m-\lceil\log_{p}(t)\rceil}-1$; (see Theorem \ref{thm2})
  \item[(3)] $r \equiv 3 (\textnormal{mod } 4)$, $N=tr+(s+1)p^{\lceil\log_{p}(t)\rceil}-2t$, for any odd $t$ and even $s$ with $1 \leq t \leq r$ and $0 \leq s \leq p^{m-\lceil\log_{p}(t)\rceil}-1$; (see Theorem \ref{thm3})
  \item[(4)] $r \equiv 3 (\textnormal{mod } 4)$, $N=tr+(s+1)p^{\lceil\log_{p}(t)\rceil}-2t+1$, for any even $t$ and even $s$ with $1 \leq t \leq r$ and $0 \leq s \leq p^{m-\lceil\log_{p}(t)\rceil}-1$; (see Theorem \ref{thm4})
  \item[(5)] $r \equiv 3 (\textnormal{mod } 4)$, $N=s(r-1)+t(r+1)-2st$ or $N=s(r-1)+t(r+1)-2st+2$, for any $0 \leq s \leq \frac{r+1}{2}$ and $0 \leq t \leq \frac{r-1}{2}$; (see Theorem \ref{thm8})
  \item[(6)] $r \equiv 3 (\textnormal{mod } 4)$, $N=s(r-1)+2t(r+1)-4st$, for any odd $s$ with $0 \leq s \leq r+1$ and $0 \leq t \leq \frac{r-1}{2}$; (see Theorem \ref{thm9} (i))
  \item[(7)] $r \equiv 3 (\textnormal{mod } 4)$, $N=s(r-1)+2t(r+1)-4st+2$, for any  $0 \leq s \leq r+1$ with $s \equiv 0 (\textnormal{mod } 4)$ and $0 \leq t \leq \frac{r-1}{2}$; (see Theorem \ref{thm9} (ii))
  \item[(8)] $r \equiv 3 (\textnormal{mod } 4)$, $N=2s(r-1)+t(r+1)-8st$ or $N=2s(r-1)+t(r+1)-8st+2$, for any $0 \leq s \leq \frac{r+1}{4}$ and $0 \leq t \leq \frac{r-1}{2}$;  (see Theorem \ref{thm10})
  \item[(9)] $r \equiv 3 (\textnormal{mod } 4)$, $N=s\frac{r-1}{2}+t(r+1)-2st$, for any even $s$ with $0 \leq s \leq r+1$ and $0 \leq t \leq \frac{r-1}{2}$; (see Theorem \ref{thm11} (i))
  \item[(10)] $r \equiv 3 (\textnormal{mod } 4)$, $N=s\frac{r-1}{2}+t(r+1)-2st+1$, for any $0 \leq s \leq r+1$ with $s \equiv 1 (\textnormal{mod } 4)$ and $0 \leq t \leq \frac{r-1}{2}$; (see Theorem \ref{thm11} (ii))
  \item[(11)] $r \equiv 3 (\textnormal{mod } 4)$, $N=s\frac{r-1}{2}+t(r+1)-2st+2$, for any $0 \leq s \leq r+1$ with $s \equiv 0 (\textnormal{mod } 4)$ and $0 \leq t \leq \frac{r-1}{2}$; (see Theorem \ref{thm11} (iii))
  \item[(12)] $r \equiv 3 (\textnormal{mod } 4)$, $N=s(r-1)+t\frac{r+1}{2}-4st$ or $N=s(r-1)+t\frac{r+1}{2}-4st+2$, for any $0 \leq s \leq \frac{r+1}{4}$ and $0 \leq t \leq \frac{r-1}{2}$; (see Theorem \ref{thm12})
  \item[(13)] $r \equiv 3 (\textnormal{mod } 4)$, $N=2s(r-1)+2t(r+1)-8st+2$, for any $0 \leq s \leq \frac{r+1}{2}$ and $0 \leq t \leq \frac{r-1}{2}$; (see Theorem \ref{thm13})
  \item[(14)] $r \equiv 3 (\textnormal{mod } 4)$, $N=s\frac{r-1}{2}+t\frac{r+1}{2}-2st$, for any even s with $0 \leq s \leq \frac{r+1}{2}$ and $0 \leq t \leq \frac{r-1}{2}$; (see Theorem \ref{thm14} (i))
  \item[(15)] $r \equiv 3 (\textnormal{mod } 4)$, $N=s\frac{r-1}{2}+t\frac{r+1}{2}-2st+1$, for any odd s with $0 \leq s \leq \frac{r+1}{2}$ and $0 \leq t \leq \frac{r-1}{2}$; (see Theorem \ref{thm14} (ii))
  \item[(16)] $r \equiv 3 (\textnormal{mod } 4)$, $N=s\frac{r-1}{2}+t\frac{r+1}{2}-2st+2$, for any even s with $0 \leq s \leq \frac{r+1}{2}$ and $0 \leq t \leq \frac{r-1}{2}$; (see Theorem \ref{thm14} (iii))
  \item[(17)] $r \equiv 3 (\textnormal{mod } 4)$, $r+1=2^{a}b$, where $b$ is odd. $N=s(r-1)+tb-2st$, for any $0 \leq s \leq b$ and even $t$ with $0 \leq t \leq r-1$; (see Theorem \ref{thm15} (i))
  \item[(18)] $r \equiv 3 (\textnormal{mod } 4)$, $r+1=2^{a}b$, where $b$ is odd. $N=s(r-1)+tb-2st+1$, for any $0 \leq s \leq b$ and odd $t$ with $0 \leq t \leq r-1$; (see Theorem \ref{thm15} (ii))
  \item[(19)] $r \equiv 3 (\textnormal{mod } 4)$, $r+1=2^{a}b$, where $b$ is odd. $N=s(r-1)+tb-2st+2$, for any $0 \leq s \leq b$ and even $t$ with $0 \leq t \leq r-1$; (see Theorem \ref{thm15} (iii))
  \item[(20)] $r \equiv 3 (\textnormal{mod } 4)$, $r+1=2^{a}b$, where $b$ is odd. $N=s(r-1)2^{a}+t(r-1)-2st$ or $N=s(r-1)2^{a}+t(r+1)-2st+2$, for any $0 \leq s \leq b$ and $0 \leq t \leq r-1$. (see Theorem \ref{thm16})
\end{description}

\subsection*{Organization of this paper}
 In Section 2, we introduce some basic notations and results about MDS Euclidean self-dual codes. In Section 3.1,  we give a construction of MDS Euclidean self-dual codes via the trace function and subspaces of $\mathbb{F}_{q}$. In Section 3.2, we present the key lemma to ensure that the symmetric difference of two intersecting subsets of finite fields can be taken as the desired evaluation set. Several new classes of MDS Euclidean self-dual codes are then constructed. In Section 4, we conclude this paper.

\section{Preliminaries}
In this section, we introduce some basic notations and results about MDS Euclidean self-dual codes.

Let $q$ be a prime power and $\mathbb{F}_{q}$ be the finite field with $q$ elements. Let $a_{1}, \ldots, a_{n}$ be $n$ distinct elements of $\mathbb{F}_{q}$ and  $v_{1}, \ldots, v_{n}$ of $\mathbb{F}_{q}$ be $n$ nonzero elements of $\mathbb{F}_{q}$. Denote $\textbf{a}= (a_{1}, \ldots, a_{n})$ and $\textbf{v}=(v_{1}, \ldots, v_{n})$. The generalized Reed-Solomon (GRS) code associated to $\textbf{a}$ and $\textbf{v}$ is defined as follows:
\[GRS_{k}(\textbf{a}, \textbf{v})\triangleq \{(v_{1}f(a_{1}), \ldots, v_{n}f(a_{n}))
 : f(x) \in \mathbb{F}_{q}[x], \textnormal{ and }\deg(f(x)) \leq k-1\} .\]
 And the extended GRS code associated to $\textbf{a}$ and $\textbf{v}$ is given as
 \[ GRS_{k}(\textbf{a}, \textbf{v},\infty) \triangleq \{(v_{1}f(a_{1}), \ldots, v_{n}f(a_{n}),f_{k-1})
            :f(x) \in \mathbb{F}_{q}[x], \textnormal{ and }\deg(f(x)) \leq k-1 \},\]
 where $f_{k-1}$ is the coefficient of $x^{k-1}$ in $f(x)$. It is well-known that (extended) GRS codes are MDS codes and so are their dual codes.

Let $\eta(x)$ be the quadratic character of $\mathbb{F}_{q}^{*}$, that is $\eta(x)=1$ if $x$ is a square in $\mathbb{F}_{q}^{*}$ and $\eta(x)=-1$ if $x$ is a non-square in $\mathbb{F}_{q}^{*}$. In this paper, we always assume that $q=r^{2}$, where $r$ is an odd prime power. Let $Tr(x)=x+x^{r}$ be the trace function from $\mathbb{F}_{q}$ to $\mathbb{F}_{r}$.

For any subset $E \subseteq \mathbb{F}_{q}$, we define the polynomial $\pi_{E}(x)$ over $\mathbb{F}_{q}$ as
\[\pi_{E}(x) \triangleq \prod_{e \in E}(x-e).\]
For any element $e \in E$, we define
\[\delta_{E}(e) \triangleq \prod_{e' \in E, e' \neq e}(e-e').\]

The Part (i) of the following lemma was given in \cite[Lemma 3.1 (1)]{ZF19} and the Part (ii) is a direct generalization of \cite[Lemma 3.1 (2)]{ZF19}, which can be proved similarly. So we omit the proof here.
\begin{lemma}\cite[Lemma 3.1]{ZF19}\label{lem1}
\begin{description}
  \item[(i)] Let $E$ be a subset of $\mathbb{F}_{q}$, then for any $e \in E$
  \[\delta_{E}(e)=\pi'_{E}(e),\]
  where $\pi'_{E}(x)$ is the derivative of $\pi_{E}(x)$.
  \item[(ii)] Let $E_{1},E_{2},\dots, E_{\ell}$ be $\ell$ pairwise disjoint subsets of $\mathbb{F}_{q}$, and $E=\bigcup_{i=1}^{\ell}E_{i}$. Then for any $e \in E_{i}$,
      \[\delta_{E}(e)=\delta_{E_{i}}(e)\prod_{1 \leq j \leq \ell, j \neq i}\pi_{E_{j}}(e).\]
\end{description}
\end{lemma}

In recent years, there are lots of work on construction of MDS Euclidean self-dual codes by using GRS codes and extended GRS codes. The key point of these constructions is the following two lemmas (or with some equivalent forms), which ensure the existence of MDS Euclidean self-dual codes. The reader may refer to \cite{JX17, Y18, FF19, ZF19, FLLL19, LLL19, FLL19} for more details.
\begin{lemma}\label{lem2}
Let $n$ be even. If there exits a subset $A \subseteq \mathbb{F}_{q}$ of size $n$, such that $\eta\big(\delta_{A}(a)\big)$ are the same for all $a \in A$, then there exists a $q$-ary MDS Euclidean self-dual code of length $n$.
\end{lemma}

\begin{lemma}\label{lem3}
Let $n$ be odd. If there exits a subset $A \subseteq \mathbb{F}_{q}$ of size $n$, such that $\eta\big(-\delta_{A}(a)\big)=1$ for all $a \in A$, then there exists a $q$-ary MDS Euclidean self-dual code of length $n+1$.
\end{lemma}

\section{New Constructions of MDS Euclidean Self-Dual Codes}
In this section, we will give several new constructions of MDS Euclidean self-dual codes based Lemmas \ref{lem2} and \ref{lem3}. The main idea is to choose different suitable subsets of $\mathbb{F}_{q}$ which satisfy the conditions in Lemmas \ref{lem2} or \ref{lem3}. Throughout this paper, we suppose that $q=r^{2}$ and $r=p^{m}$, where $p$ is an odd prime.

\subsection{MDS Euclidean Self-Dual Codes from Trace Function and Subspaces}
In this subsection, we provide our first construction of MDS Euclidean self-dual codes, which is based on the trace function and a subspace of $\mathbb{F}_{r}$.

Let $1 \leq t \leq r$ and $s$ be even with $0 \leq s \leq p^{m-t'}-1$, where $t'=\lceil\log_{p}(t)\rceil$.  We fix an $\mathbb{F}_{p}$-linear subspace $H \subseteq \mathbb{F}_{r}$ of dimensional $t'$. Then $|H|=p^{t'} \geq t$ and $|\mathbb{F}_{r}/H |=p^{m-t'}>s$. Let $h_{1}=0, h_{2}, \ldots, h_{t}$ be $t$ distinct elements of $H$. Let $b_{0}=0, b_{1}, b_{2}, \ldots, b_{s}$ be $s+1$ distinct representations of $\mathbb{F}_{r}/H$ such that for any $1 \leq i \leq \frac{s}{2}$,
\[b_{i}=-b_{\frac{s}{2}+i}.\]
For any $1 \leq i \leq t$, define
\begin{equation}\label{1}
  T_{i} \triangleq \{x \in \mathbb{F}_{q} : Tr(x)=h_{i} \}.
\end{equation}
Then $|T_{i}|=r$ and $T_{i} \bigcap T_{j}= \emptyset$ for any $i \neq j$.
\vskip 1mm \noindent
For any $0 \leq j \leq s$, define
\begin{equation}\label{2}
  H_{j} \triangleq \{b_{j}+h: h \in H \}.
\end{equation}
Then each $H_j$ is a subset of $\mathbb{F}_{r}$ with $|H_{j}|=p^{t'}$ and $H_{i} \bigcap H_{j}= \emptyset$ for any $0 \leq i \neq j \leq s$.
\begin{lemma}\label{lem4}
With the above notations,
\begin{description}
  \item[(i)] $\pi_{T_{i}}(x)=Tr(x)-h_{i}=x+x^{r}-h_{i}$, and  $\pi'_{T_{i}}(x)=1.$
  \item[(ii)] For any $1 \leq i \leq t$ and $1 \leq j \leq s$,
  \[T_{i} \cap H_{j}= \emptyset,\]
  and
  \[T_{i} \cap H_{0}= \{\frac{h_{i}}{2} \}.\]
\end{description}
\end{lemma}
\begin{proof}
\begin{description}
  \item[(i)] The conclusions follow immediately from the definitions.
  \item[(ii)] Let $x \in T_{i}\cap \mathbb{F}_{r}$, then $x+x^{r}=x+x=h_{i}$, i.e., $x=\frac{h_{i}}{2} \in H_{0}$. The conclusions then follow from $H_{0} \bigcap H_{j}= \emptyset$.
\end{description}
\end{proof}

Based on \eqref{1} and \eqref{2}, we give our first construction as follows.
\begin{theorem}\label{thm1}
Suppose $q=r^{2}$ and $r=p^{m}$. For any even $t$ and even $s$ with $1 \leq t \leq r$ and $0 \leq s \leq p^{m-t'}-1$, let $n=tr+sp^{t'}$, where $t'=\lceil\log_{p}(t)\rceil$, then there exists a $q$-ary MDS Euclidean self-dual code of length $n$.
\end{theorem}

\begin{proof}
Let $T_{i}$ and $H_{j}$ be defined as \eqref{1} and \eqref{2}, respectively. Define
\[A=(\bigcup_{i=1}^{t}T_{i})\cup (\bigcup_{j=1}^{s}H_{j}).\]
From Lemma \ref{lem4} (ii), these $T_{i}$ and $H_{j}$ are pairwise disjoint. Thus $|A|=tr+sp^{t'}=n$.
We begin to calculate $\delta_{A}(a)$ for any $a \in A$.
\vskip 2mm
\textbf{i)} If $a \in T_{i}$ for some $i$, then by Lemma \ref{lem1} and Lemma \ref{lem4} (i),
\begin{eqnarray*}
  \delta_{A}(a) &=& \delta_{T_{i}}(a)\left(\prod_{\ell \neq i,\ell=1}^{t}\pi_{T_{\ell}}(a)\right)\left(\prod_{j=1}^{s}\pi_{H_{j}}(a)\right) \\
   &=& \pi'_{T_{i}}(a)\left(\prod_{\ell \neq i,\ell=1}^{t}\pi_{T_{\ell}}(a)\right)\left(\prod_{j=1}^{s}\pi_{H_{j}}(a)\right) \\
   &=& \left(\prod_{\ell \neq i,\ell=1}^{t}(Tr(a)-h_{\ell})\right)\left(\prod_{j=1}^{s}\pi_{H_{j}}(a)\right).
\end{eqnarray*}
 Note that $\pi_{H_{j}}(a)=\prod_{\xi \in H_{j}}(a-\xi)$ and $a^{r}+a=h_{i}$, thus
  \begin{eqnarray*}
    \pi_{H_{j}}^{r}(a) &=& \prod_{\xi \in H_{j}}(a^{r}-\xi) \\
                       &=& \prod_{\xi \in H_{j}}(h_{i}-a-\xi) \\
                       &=& (-1)^{|H_{j}|}\prod_{\xi \in H_{j}}\big(a-(h_{i}-\xi)\big).
  \end{eqnarray*}
 Note that $b_{i}=-b_{\frac{s}{2}+i}$, thus $h_{i}-\xi$ runs over $H_{\frac{s}{2}+j}$ when $\xi$ runs over $H_{j}$. Hence
   \begin{eqnarray*}
    \pi_{H_{j}}^{r}(a) &=& -\prod_{\xi \in H_{\frac{s}{2}+j}}(a-\xi) \\
                       &=& -\pi_{H_{\frac{s}{2}+j}}(a).
  \end{eqnarray*}

 Therefore,
  \[\big(\prod_{j=1}^{s}\pi_{H_{j}}(a)\big)^{r}=(-1)^{s}\prod_{j=1}^{s}\pi_{H_{j}}(a)=\prod_{j=1}^{s}\pi_{H_{j}}(a).\]
  Hence $\prod_{j=1}^{s}\pi_{H_{j}}(a) \in \mathbb{F}_{r}$. And each $Tr(a)-h_{i} \in \mathbb{F}_{r}$, thus $\delta_{A}(a) \in \mathbb{F}_{r}$.
  \vskip 2mm
\textbf{ii)} If $a \in H_{j}$ for some $j$, then by Lemma \ref{lem1} and Lemma \ref{lem4} (i) again,
  \[\delta_{A}(a)=\big(\prod_{i=1}^{t}\pi_{T_{i}}(a)\big)\cdot\pi'_{H_{j}}(a)\prod_{\ell \neq j, \ell=1}^{s}\pi_{H_{\ell}}(a).\]
  Since each $H_{\ell} \subseteq \mathbb{F}_{r}$ and $a \in H_{j}$, we have $\pi'_{H_{j}}(a), \pi_{H_{\ell}}(a) \in \mathbb{F}_{r}$,
  thus $\delta_{A}(a) \in \mathbb{F}_{r}$.

  In a word, we have $\delta_{A}(a) \in \mathbb{F}_{r}$ for all $a \in A$. Thus $\eta\big(\delta_{A}(a)\big)=1$ since each element of $\mathbb{F}_{r}$ is a square in $\mathbb{F}_{q}$. The theorem then follows from Lemma \ref{lem2}.
\end{proof}

\begin{example}
In Theorem \ref{thm1}, let $p=5$, $r=p^{2}=25$, $q=r^{2}=625$, $t=4$. Then $t'=\lceil\log_{p}(t)\rceil=1$. Let $s=2$, then $n=tr+sp=110$. So we can obtain an MDS Euclidean self-dual code of length 110 over $\mathbb{F}_{625}$. The parameters of the code are new in the sense that they have not been obtained in the literature (see Table I).
\end{example}
Based on the construction of Theorem \ref{thm1}, we can provide the following theorem by Lemma \ref{lem3}.
\begin{theorem}\label{thm2}
Suppose $q=r^{2}$ and $r=p^{m}$. For any odd $t$ and even $s$ with $1 \leq t \leq r$ and $0 \leq s \leq p^{m-\lceil\log_{p}(t)\rceil}-1$, let $n=tr+sp^{\lceil\log_{p}(t)\rceil}$, then there exists a $q$-ary MDS Euclidean self-dual code of length $n+1$.
\end{theorem}

\begin{proof}
Let $A$ be defined as in the proof of Theorem \ref{thm1}. With the same argument of Theorem \ref{thm1}, we can still prove that $\delta_{A}(a) \in \mathbb{F}_{r}$ for all $a \in A$. Hence $\eta(-\delta_{A}(a))=1$. Note that $t$ is odd and $s$ is even, hence $n$ is odd. The theorem then follows from Lemma \ref{lem3}.
\end{proof}

\begin{example}
In Theorem \ref{thm2}, let $p=3$, $r=p^{3}=27$, $q=r^{2}=729$, $t=9$. Then $\lceil\log_{p}(t)\rceil=1$. Let $s=8$, then $n=tr+sp=267$. So we can obtain an MDS Euclidean self-dual code of length $268$ over $\mathbb{F}_{729}$. The parameters of the code are new in the sense that they have not been obtained in the literature (see Table I).
\end{example}

\subsection{MDS Euclidean Self-Dual Codes from Two Intersecting Subsets}
In this subsection, we give a sufficient condition under which the symmetric difference of two intersecting sets satisfies the conditions in Lemmas \ref{lem2} or \ref{lem3}.

For any two sets $A$ and $B$,  their \emph{difference }is defined as
\[A\setminus B \triangleq \{a \mid a \in A \textnormal{ and } a \notin B \}, \]
and their \emph{symmetric difference} is defined as
\[A \triangle B \triangleq (A\cup B)\setminus(A\cap B)=(A\setminus B)\cup(B\setminus A).\]
It is obviously that $|A\triangle B|=|A|+|B|-2|A\cap B|.$
\vskip 2mm
The following lemma is the key lemma for our next constructions.

\begin{lemma}\label{lem5}
Fix a constant $c \in \{1, -1\} \subseteq \mathbb{F}_{q}$ and let $A, B$ be two subsets of $\mathbb{F}_{q}$. If for all $a \in A\setminus B$ and $b \in B\setminus A$,
\[\eta\big(\delta_{A}(a)\pi_{B}(a)\big)=\eta\big(\pi_{A}(b)\delta_{B}(b)\big)=c,\]
 then,
for all $e \in A\triangle B$,
\[\eta\big(\delta_{A \triangle B}(e)\big)=c.\]
\end{lemma}
\begin{proof}
When $e \in A\setminus B$, then by Lemma \ref{lem1} and definitions,
\begin{eqnarray*}
  \delta_{A \triangle B}(e) &=& \delta_{A\setminus B}(e)\pi_{B\setminus A}(e) \\
    &=& \prod_{a \in A\setminus B, a \neq e}(e-a)\prod_{b \in B\setminus A}(e-b) \\
    &=& \frac{\prod_{a \in A, a \neq e}(e-a)}{\prod_{e' \in A \cap B}(e-e')}\prod_{b \in B\setminus A}(e-b)  \\
    &=& \frac{\prod_{a \in A, a \neq e}(e-a)}{\prod_{e' \in A \cap B}(e-e')^{2}}\prod_{b \in B}(e-b)  \\
    &=&  \frac{\delta_{A}(e)\pi_{B}(e)}{\pi^{2}_{A\cap B}(e)}.
\end{eqnarray*}
Hence, $\eta\big(\delta_{A \triangle B}(e)\big)=\eta\big(\delta_{A}(e)\pi_{B}(e)\big)=c$. Similarly, it holds for $e \in B\setminus A$. The lemma is proved.
\end{proof}

According to Lemma \ref{lem5}, we provide the following construction based on the two sets given in \eqref{1} and \eqref{2}.
\begin{theorem}\label{thm3}
Suppose $q=r^{2}$ and $r=p^{m}$ with $r \equiv 3 (\textnormal{mod } 4)$. For any odd $t$ and even $s$ with $1 \leq t \leq r$ and $0 \leq s \leq p^{m-\lceil\log_{p}(t)\rceil}-1$, let $n=tr+(s+1)p^{\lceil\log_{p}(t)\rceil}-2t$, then there exists a $q$-ary MDS Euclidean self-dual code of length $n$.
\end{theorem}

\begin{proof}
Let
\[A=\bigcup_{i=1}^{t}T_{i}, B=\bigcup_{j=0}^{s}H_{j},\]
then according to Lemma \ref{lem4},
\[A \cap B= \{\frac{h_{1}}{2}, \frac{h_{2}}{2}, \ldots, \frac{h_{t}}{2}\},\]
hence $|A \cap B|=t$ and $|A\bigtriangleup B|=|A|+|B|-2|A \cap B|=tr+(s+1)p^{\lceil\log_{p}(t)\rceil}-2t=n$.

For any $a \in A\setminus B$, suppose $a \in T_{i}$ for some $i$. Then
\[\delta_{A}(a)\pi_{B}(a)=\pi'_{T_{i}}(a)\prod_{\ell \neq i, \ell=1}^{t}\pi_{T_{\ell}}(a)\prod_{j=0}^{s}\pi_{H_{j}}(a)=\prod_{\ell \neq i, \ell=1}^{t}\pi_{T_{\ell}}(a)\prod_{j=0}^{s}\pi_{H_{j}}(a).\]
Similarly, we can show that
\[\big(\prod_{j=0}^{s}\pi_{H_{j}}(a)\big)^{r}=(-1)^{s+1}\prod_{j=0}^{s}\pi_{H_{j}}(a)=-\prod_{j=0}^{s}\pi_{H_{j}}(a),\]
i.e.,
\[\big(\prod_{j=0}^{s}\pi_{H_{j}}(a)\big)^{r-1}=-1.\]
Then there exists an odd integer $k$ and a primitive element $\omega$ of $\mathbb{F}_{q}$, such that
\[\prod_{j=0}^{s}\pi_{H_{j}}(a)=\omega^{\frac{(r+1)}{2}k}.\]
Since $r \equiv 3 (\textnormal{mod } 4)$, $\frac{r+1}{2}$ is even, thus $\prod_{j=0}^{s}\pi_{H_{j}}(a)$ is a square in $\mathbb{F}_{q}$. Note that $\pi_{T_{\ell}}(a) \in \mathbb{F}_{r}$, which is also a square in $\mathbb{F}_{q}$. Thus $\eta\big(\delta_{A}(a)\pi_{B}(a)\big)=1$.

For any $b \in B\setminus A$, suppose $b \in H_{j}$ for some $j$. Then
\[\pi_{A}(b)\delta_{B}(b)=\prod_{i=1}^{t}\pi_{T_{i}}(b)\pi'_{H_{j}}(b)\prod_{\ell \neq j, \ell=0}^{s}\pi_{H_{j}}(b) \in \mathbb{F}_{r}.\]
Hence $\eta\big(\pi_{A}(b)\delta_{B}(b)\big)=1$. By Lemma \ref{lem5}, for all $e \in A\triangle B$,
\[\eta\big(\delta_{A \triangle B}(e)\big)=1.\]
The theorem then follows from Lemma \ref{lem2}.
\end{proof}

\begin{example}
In Theorem \ref{thm3}, let $p=r=23$, $q=r^{2}=529$, $t=5$. Then $t'=\lceil\log_{p}(t)\rceil=1$. Let $s=0$, then $n=tr+(s+1)p-2t=128$. So we can obtain an MDS Euclidean self-dual code of length $128$ over $\mathbb{F}_{529}$. The parameters of the code are new in the sense that they have not been obtained in the literature (see Table I).
\end{example}

Similarly, based on Theorem \ref{thm3} and Lemma \ref{lem3}, we can obtain the following theorem.

\begin{theorem}\label{thm4}
Suppose $q=r^{2}$ and $r=p^{m}$ with $r \equiv 3 (\textnormal{mod } 4)$. For any even $t$ and even $s$ with $1 \leq t \leq r$ and $0 \leq s \leq p^{m-\lceil\log_{p}(t)\rceil}-1$, let $n=tr+(s+1)p^{\lceil\log_{p}(t)\rceil}-2t$, then there exists a $q$-ary MDS Euclidean self-dual code of length $n+1$.
\end{theorem}

\begin{proof}
Let $A$ and $B$ are defined as in Theorem \ref{thm3}. Similar to the proof of Theorem \ref{thm3}, for all $e \in A\triangle B$, we still have
\[\eta\big(\delta_{A \triangle B}(e)\big)=1,\]
i.e.,
\[\eta\big(-\delta_{A \triangle B}(e)\big)=1.\]
Since both $t$ and $s$ are even, $n$ is odd. The theorem then follows from Lemma \ref{lem3}.
\end{proof}

\begin{example}
In Theorem \ref{thm4}, let $p=3$, $r=3^{3}=27$, $q=r^{2}=729$, $t=6$. Then $t'=\lceil\log_{p}(t)\rceil=2$. Let $s=2$, then $n=tr+(s+1)p-2t=159$. So we can obtain an MDS Euclidean self-dual code of length $160$ over $\mathbb{F}_{729}$. The parameters of the code are new in the sense that they have not been obtained in the literature (see Table I).
\end{example}

In the following, we consider two multiplicative subgroups of $\mathbb{F}^{*}_{q}$ and their cosets. Let $\omega$ be a primitive element of $\mathbb{F}_{q}$, and let \[\alpha=\omega^{\mu} \textnormal{ and } \beta=\omega^{\nu},\]
where $\mu$ and $\nu$ are two distinct divisors of $q-1$. Let $\langle\alpha \rangle$ and $\langle\beta \rangle$ be the two multiplicative subgroups of $\mathbb{F}^{*}_{q}$ generated by $\alpha$ and $\beta$, respectively. Suppose $0 \leq s \leq \frac{\mu}{\gcd(\mu,\nu)}$ and $0 \leq t \leq \frac{\nu}{\gcd(\mu,\nu)}$.  Denote $A_{i}=\beta^{i}\langle\alpha \rangle$ and $B_{j}=\alpha^{i}\langle\beta \rangle$. Let
\begin{equation}\label{3}
  A \triangleq \bigcup_{i=0}^{s-1}A_{i},
\end{equation}
and
\begin{equation}\label{4}
  B\triangleq \bigcup_{j=0}^{t-1}B_{j}.
\end{equation}

\begin{lemma}\label{lem6}
Keep the notations as above. Let $A$ and $B$ be defined as \eqref{3} and \eqref{4}, respectively. Then
\[|A|=s\frac{q-1}{\mu}, |B|=t\frac{q-1}{\nu},\]
and
\[|A \cap B|=\frac{(q-1)\gcd(\mu,\nu)}{\mu\nu}st.\]
\end{lemma}

\begin{proof}
To prove $|A|=s\frac{q-1}{\mu}$, it is only need to show that $1, \beta, \ldots, \beta^{s-1}$ are the representations of $s$ distinct cosets of the subgroup $\langle \alpha \rangle$ in $\mathbb{F}_{q}^{*}$. By contradiction, suppose there exist $0 \leq i_{1} < i_{2} \leq s-1$ and $0 \leq j \leq \frac{q-1}{\mu}-1$ such that $\beta^{i_{2}}=\alpha^{j}\beta^{i_{1}}$. Denote $i =i_{2}-i_{1}$, then  $1 \leq i \leq s-1 < \frac{\mu}{\gcd(\mu,\nu)}$. Note that $\alpha=\omega^{\mu}$ and $\beta=\omega^{\nu}$, we have $\omega^{\nu i-\mu j}=1$. Thus $(q-1) \mid (\nu i-\mu j)$, which implies that $\mu \mid \nu i$.  Thus $\frac{\mu}{\gcd(\mu,\nu)} \mid i $, a contradiction. The first conclusion follows. Similarly, we can prove that $|B|=t\frac{q-1}{\nu}$.

Now let $e \in A\cap B$, then there exist some $0 \leq i_{1} \leq s-1$, $0 \leq j_{1} \leq \frac{q-1}{\mu}-1$, $0 \leq i_{2} \leq t-1$ and $0 \leq j_{2} \leq \frac{q-1}{\nu}-1$ such that \[e=\beta^{i_{1}}\alpha^{j_{1}}=\alpha^{i_{2}}\beta^{j_{2}}.\]
Thus to prove $|A \cap B|=\frac{(q-1)\gcd(\mu,\nu)}{\mu\nu}st$, we only need to show that given $0 \leq i_{1} \leq s-1$ and $0 \leq i_{2} \leq t-1$, the number of pairs $(j_{1}, j_{2})$ with $0 \leq j_{1} \leq \frac{q-1}{\mu}-1$ and $0 \leq j_{2} \leq \frac{q-1}{\nu}-1$, which satisfy
\begin{equation}\label{5}
  \beta^{i_{1}}\alpha^{j_{1}}=\alpha^{i_{2}}\beta^{j_{2}}
\end{equation}
is $\frac{(q-1)\gcd(\mu,\nu)}{\mu\nu}$. Indeed, \eqref{5} is equivalent to
\begin{equation}\label{6}
  \omega^{\mu(j_{1}-i_{2})+\nu(i_{1}-j_{2})}=1,
\end{equation}
hence $\nu \mid \mu(j_{1}-i_{2})$, i.e., $\frac{\nu}{\gcd(\mu,\nu)} \mid (j_{1}-i_{2})$.  Thus $j_{2} (\textnormal{mod } r+1)$ is unique. Since $\mu, \nu \mid (q-1)$, it is not hard to prove that $\frac{\mu\nu}{\gcd(\mu,\nu)} \mid (q-1)$. Thus for $0 \leq j_{1} \leq \frac{q-1}{\mu}-1$, the number of $j_{1}$ satisfying $\frac{\nu}{\gcd(\mu,\nu)} \mid (j_{1}-i_{2})$ is $\frac{q-1}{\mu}/\frac{\nu}{\gcd(\mu,\nu)}=\frac{(q-1)\gcd(\mu,\nu)}{\mu\mu}$. According to \eqref{6} and $0 \leq j_{2} \leq \frac{q-1}{\nu}-1$, $j_{2}$ is uniquely determined after fixing $i_{1}, i_{2}, j_{1}$. The lemma is proved.
\end{proof}

Based on the two sets $A$ and $B$ defined by \eqref{3} and \eqref{4}, respectively, we provide a generic construction of MDS Euclidean self-dual codes as follows.
\begin{theorem}\label{thm5}
Let $q=r^{2}$ and $r \equiv 3 (\textnormal{mod } 4)$. Suppose $\mu, \nu \mid (q-1)$. Let $0 \leq s \leq \frac{\mu}{\gcd(\mu, \nu)}$ and $0 \leq t \leq \frac{\nu}{\gcd(\mu, \nu)}$. Put
\[n=s\frac{q-1}{\mu}+t\frac{q-1}{\nu}-\frac{2(q-1)\gcd(\mu,\nu)}{\mu\nu}st.\]
 Suppose the following conditions hold:
\begin{description}
  \item[(i)] $n$ is even;
  \item[(ii)] $\mu$ is even, $\mu \mid \nu(r+1)$ and $\nu \mid \mu(r-1)$;
  \item[(iii)] both $(\frac{q-1}{\mu}-1)\nu-\frac{(r+1)\nu}{\mu}s$ and $\frac{(r+1)\nu}{\mu}s+\nu$ are even.
\end{description}
Then there exists a $q$-ary MDS Euclidean self-dual code of length $n$.
\end{theorem}

\begin{proof}
Let $\omega$ be a primitive element of $\mathbb{F}_{q}$ and $\alpha=\omega^{\mu}, \beta=\omega^{\nu}$. Denote $A_{i}=\beta^{i}\langle\alpha \rangle$ and $B_{j}=\alpha^{i}\langle\beta \rangle$. Let $A$ and $B$ be defined by \eqref{3} and \eqref{4}, respectively. Then by Lemma \ref{lem6}, we have
\[|A\triangle B|=|A|+|B|-2|A\cap B|=s\frac{q-1}{\mu}+t\frac{q-1}{\nu}-\frac{2(q-1)\gcd(\mu,\nu)}{\mu\nu}st=n.\]
For any  $0 \leq i \leq s-1$ and $0 \leq j \leq t-1$, we have
\[\pi_{A_{i}}(x)=\prod_{\ell=0}^{\frac{q-1}{\mu}-1}(x-\beta^{i}\alpha^{\ell})=x^{\frac{q-1}{\mu}}-\beta^{\frac{q-1}{\mu}i},\]
\[\pi_{A_{i}}'(x)=\frac{q-1}{\mu}x^{\frac{q-1}{\mu}-1},\]
\[\pi_{B_{j}}(x)=\prod_{\ell=0}^{\frac{q-1}{y}-1}(x-\alpha^{j}\beta^{\ell})=x^{\frac{q-1}{y}}-\alpha^{\frac{q-1}{y}j},\]
and
\[\pi_{B_{j}}'(x)=\frac{q-1}{y}x^{\frac{q-1}{y}-1}.\]
On the one hand, for any $\beta^{i}\alpha^{j} \in A\setminus B$, where $0 \leq i \leq s-1$ and $0 \leq j \leq r-2$, by Lemma \ref{lem1}, we have
\begin{eqnarray*}
  \delta_{A}(\beta^{i}\alpha^{j}) &=& \pi'_{A_{i}}(\beta^{i}\alpha^{j})\prod_{\ell \neq i, \ell=0}^{t-1}\pi_{A_{\ell}}(\beta^{i}\alpha^{j}) \\
    &=& \frac{q-1}{\mu}\beta^{(\frac{q-1}{\mu}-1)i}\alpha^{-j}\prod_{\ell \neq i, \ell=0}^{s-1}(\beta^{\frac{q-1}{\mu}i}- \beta^{\frac{q-1}{\mu}\ell})\\
    &=& \frac{q-1}{\mu}\omega^{(\frac{q-1}{\mu}-1)\nu i-\mu j}\prod_{\ell \neq i, \ell=0}^{s-1}(\omega^{\frac{q-1}{\mu}\nu i}- \omega^{\frac{q-1}{\mu}\nu \ell}).
\end{eqnarray*}
Denote $\Omega=\prod_{\ell \neq i, \ell=0}^{s-1}(\omega^{\frac{q-1}{\mu}\nu i}- \omega^{\frac{q-1}{\mu}\nu \ell})$. Since $\mu \mid (r+1)\nu$, \[(\omega^{\frac{q-1}{\mu}\nu i})^{r+1}=(\omega^{\frac{(r+1)\nu}{\mu} i})^{q-1}=1,\]
i.e.,
\[(\omega^{\frac{q-1}{\mu}\nu i})^{r}=\omega^{-\frac{q-1}{\mu}\nu i}.\]
Thus
\begin{eqnarray*}
  \Omega^{r} &=& \prod_{\ell \neq i, \ell=0}^{s-1}\big((\omega^{\frac{q-1}{\mu}\nu i})^{r}- (\omega^{\frac{q-1}{\mu}\nu \ell})^{r}\big) \\
    &=& \prod_{\ell \neq i, \ell=0}^{s-1}(\omega^{-\frac{q-1}{\mu}\nu i}- \omega^{-\frac{q-1}{\mu}\nu \ell}) \\
    &=& \prod_{\ell \neq i, \ell=0}^{s-1}\omega^{-\frac{q-1}{\mu}\nu (i+\ell)}(\omega^{\frac{q-1}{\mu}\nu \ell}- \omega^{\frac{q-1}{\mu}\nu i}) \\
    &=& (-1)^{s-1} \omega^{-\frac{q-1}{\mu}\nu\big((s-2)i+\frac{s(s-1)}{2}\big)}\Omega.
\end{eqnarray*}
Hence $\Omega^{r-1}=(-1)^{s-1} \omega^{-\frac{q-1}{\mu}\nu\big((s-2)i+\frac{s(s-1)}{2}\big)}$. Note that $-1=\omega^{(r-1)\frac{r+1}{2}}$ and $\frac{q-1}{\mu}\nu=(r-1)\frac{\nu(r+1)}{\mu}$, thus there exists an integer $k$, such that
\[\Omega=\omega^{\frac{r+1}{2}(s-1)-\frac{\nu(r+1)}{\mu}\big((s-2)i+\frac{s(s-1)}{2}\big)+k(r+1)}.\]
Hence,
\begin{equation}\label{7}
  \delta_{A}(\beta^{i}\alpha^{j})=\frac{q-1}{\mu}\omega^{(\frac{q-1}{\mu}-1)\nu i-\mu j+\frac{r+1}{2}(s-1)-\frac{\nu(r+1)}{\mu}\big((s-2)i+\frac{s(s-1)}{2}\big)+k(r+1)}.
\end{equation}
$\frac{q-1}{\mu}$ is a square in $\mathbb{F}_{q}$ since each element of $\mathbb{F}_{p}$ is a square in $\mathbb{F}_{q}$. Since $r \equiv 3 (\textnormal{mod } 4)$, $\frac{r+1}{2}$ is even. By conditions (ii) and (iii), $\mu$ and $(\frac{q-1}{\mu}-1)\nu-\frac{(r+1)\nu}{\mu}s$ are even. Thus, from Eq. \eqref{7}, we have
\[\eta(\delta_{A}(\beta^{i}\alpha^{j}))=\eta(\omega^{\frac{\nu(r+1)}{\mu}\frac{s(s-1)}{2}}).\]
Note that
\begin{eqnarray*}
  \pi_{B}(\beta^{i}\alpha^{j}) &=& \prod_{\ell=0}^{t-1}\pi_{B_{\ell}}(\beta^{i}\alpha^{j}) \\
    &=& \prod_{\ell=0}^{t-1}(\alpha^{\frac{q-1}{\nu}j}-\alpha^{\frac{q-1}{\nu}\ell})\\
    &=& \prod_{\ell=0}^{t-1}(\omega^{\frac{q-1}{\nu}\mu j}-\omega^{\frac{q-1}{\nu}\mu\ell}).
\end{eqnarray*}
Since $\nu \mid \mu(r-1)$,
\[(\omega^{\frac{q-1}{\nu}\mu})^{r-1}=(\omega^{\frac{\mu(r-1)}{\nu}})^{q-1}=1.\]
Thus $\omega^{\frac{q-1}{\nu}\mu} \in \mathbb{F}_{r}$, hence $\pi_{B}(\beta^{i}\alpha^{j}) \in \mathbb{F}_{r},$ which is a square in $\mathbb{F}_{q}$.
Therefore,
\[\eta\big(\delta_{A}(\beta^{i}\alpha^{j})\pi_{B}(\beta^{i}\alpha^{j})\big)=\eta(\omega^{\frac{\nu(r+1)}{\mu}\frac{s(s-1)}{2}}).\]
On the other hand, for any $\alpha^{i}\beta^{j} \in B\setminus A$,
\begin{eqnarray*}
  \pi_{A}(\alpha^{i}\beta^{j}) &=& \prod_{\ell=0}^{s-1}\pi_{A_{\ell}}(\alpha^{i}\beta^{j}) \\
   &=& \prod_{\ell=0}^{s-1}(\beta^{\frac{q-1}{\mu}j}- \beta^{\frac{q-1}{\mu}\ell})\\
   &=& \prod_{\ell=0}^{s-1}(\omega^{\frac{q-1}{\mu}\nu j}- \omega^{\frac{q-1}{\mu}\nu \ell}).
\end{eqnarray*}
Similarly as above, we can show that there exists an integer $k$ such that
\begin{equation}\label{8}
  \pi_{A}(\alpha^{i}\beta^{j})=\omega^{\frac{r+1}{2}s-\frac{\nu(r+1)}{\mu}\big(sj+\frac{s(s-1)}{2}\big)+k(r+1)}.
\end{equation}
By Lemma \ref{lem1} again,
\begin{equation} \label{9}
  \begin{split}
  \delta_{B}(\alpha^{i}\beta^{j})&=\pi_{B_{i}}'(\alpha^{i}\beta^{j})\prod_{\ell \neq i}\pi_{B_{\ell}}(\alpha^{i}\beta^{j})    \\
  &=\frac{q-1}{\nu}\alpha^{(\frac{q-1}{\nu}-1)i}\beta^{-j}\prod_{\ell \neq i}(\alpha^{\frac{q-1}{\nu}i}-\alpha^{\frac{q-1}{\nu}\ell})\\
  &= \frac{q-1}{\nu}\omega^{(\frac{q-1}{\nu}-1)\mu i-\nu j}\prod_{\ell \neq i}(\omega^{\frac{q-1}{\nu}\mu j}-\omega^{\frac{q-1}{\nu}\mu\ell}).
  \end{split}
\end{equation}
By the above proof, we know that $\prod_{\ell \neq i}(\omega^{\frac{q-1}{\nu}\mu j}-\omega^{\frac{q-1}{\nu}\mu\ell}) \in \mathbb{F}_{r}$ which is a square in $\mathbb{F}_{q}$. Thus by conditions (ii) and (iii), we can deduce that
\[\eta\big(\pi_{A}(\alpha^{i}\beta^{j})\delta_{B}(\alpha^{i}\beta^{j})\big)=\eta(\omega^{\frac{\nu(r+1)}{\mu}\frac{s(s-1)}{2}}).\]
In summary, by Lemma \ref{lem5}, for any $e \in A\bigtriangleup B$,
\[\eta\big(\delta_{A\bigtriangleup B}(e)\big)=\eta(\omega^{\frac{\nu(r+1)}{\mu}\frac{s(s-1)}{2}}).\]
The theorem then follows from Lemma \ref{lem2}.
\end{proof}

Based on Theorem \ref{thm5}, we can similarly give a construction of MDS Euclidean self-dual codes with length $n+1$ by Lemma \ref{lem3} as follows.
\begin{theorem}\label{thm6}
Let $q=r^{2}$ and $r \equiv 3 (\textnormal{mod } 4)$. Suppose $\mu, \nu \mid (q-1)$. Let $0 \leq s \leq \frac{\mu}{\gcd(\mu, \nu)}$ and $0 \leq t \leq \frac{\nu}{\gcd(\mu, \nu)}$. Put
\[n=s\frac{q-1}{\mu}+t\frac{q-1}{\nu}-\frac{2(q-1)\gcd(\mu,\nu)}{\mu\nu}st.\]
Suppose the following conditions hold:
\begin{description}
  \item[(i)] $n$ is odd;
  \item[(ii)] $\mu$ is even, $\mu \mid \nu(r+1)$ and $\nu \mid \mu(r-1)$;
  \item[(iii)] both $(\frac{q-1}{\mu}-1)\nu-\frac{(r+1)\nu}{\mu}s$ and $\frac{(r+1)\nu}{\mu}s+\nu$ are even;
  \item[(iv)] $\frac{\nu(r+1)}{\mu}\frac{s(s-1)}{2}$ is even.
\end{description}
Then there exists a $q$-ary MDS Euclidean self-dual code of length $n+1$.
\end{theorem}

\begin{proof}
Let $A$ and $B$ be defined as in Theorem \ref{thm5}. Then from the proof of Theorem \ref{thm5}, for any $e \in A\bigtriangleup B$,
\[\eta\big(\delta_{A\bigtriangleup B}(e)\big)=\eta(\omega^{\frac{\nu(r+1)}{\mu}\frac{s(s-1)}{2}})=1.\]
The second equality follows by condition (iv). Since $\eta(-1)=1$, we have
\[\eta\big(-\delta_{A\bigtriangleup B}(e)\big)=1.\]
The theorem then follows from Lemma \ref{lem3}.
\end{proof}

 Modifying the conditions of Theorem \ref{thm5} slightly and adding the zero element into the set $A$, we provide the following construction of MDS Euclidean self-dual codes with length $n+2$ via  Lemma \ref{lem3}.
\begin{theorem}\label{thm7}
Let $q=r^{2}$ and $r \equiv 3 (\textnormal{mod } 4)$. Suppose $\mu, \nu \mid (q-1)$. Let $0 \leq s \leq \frac{\mu}{\gcd(\mu, \nu)}$ and $0 \leq t \leq \frac{\nu}{\gcd(\mu, \nu)}$. Put
\[n=s\frac{q-1}{\mu}+t\frac{q-1}{\nu}-\frac{2(q-1)\gcd(\mu,\nu)}{\mu\nu}st.\]
Suppose the following conditions hold:
\begin{description}
  \item[(i)] $n$ is even;
  \item[(ii)] $\mu \mid \nu(r+1)$ and $\nu \mid \mu(r-1)$;
  \item[(iii)] both $\frac{q-1}{\mu}\nu$ and $\frac{(r+1)\nu}{\mu}s$ are even;
  \item[(iv)] $\frac{\nu(r+1)}{\mu}\frac{s(s-1)}{2}$ is even.
\end{description}
Then there exists a $q$-ary MDS Euclidean self-dual code of length $n+2$.

\end{theorem}

\begin{proof}
Let $A$ and $B$ be defined as in Theorem \ref{thm5}. Let $A'=\{0\} \cup A$. Then $|A'\bigtriangleup B|=1+|A\bigtriangleup B|=n+1$. Firstly,
\begin{eqnarray*}
  \delta_{A'}(0)\pi_{B}(0) &=& \prod_{i=0}^{s-1}\pi_{A_{i}}(0)\prod_{j=0}^{t-1}\pi_{B_{j}}(0) \\
    &=&  (-1)^{s+t}\prod_{i=0}^{s-1}\beta^{\frac{q-1}{\mu}i}\prod_{j=0}^{t-1}\alpha^{\frac{q-1}{\nu}j} \\
    &=&  (-1)^{s+t}\prod_{i=0}^{s-1}\prod_{j=0}^{t-1}\omega^{\frac{q-1}{\mu}\nu i+\frac{q-1}{\nu}\mu j}.
\end{eqnarray*}
Both $\frac{q-1}{\mu}\nu=(r+1)\frac{\mu(r-1)}{\nu}$ and $\frac{q-1}{\nu}\mu=(r-1)\frac{\nu(r+1)}{\mu}$ are even by condition (i). Thus
\[\eta\big(\delta_{A'}(0)\pi_{B}(0)\big)=1.\]
Secondly, for any $\beta^{i}\alpha^{j} \in A'\setminus B$,
from \eqref{7},
\begin{eqnarray*}
  \delta_{A'}(\beta^{i}\alpha^{j}) &=& \beta^{i}\alpha^{j}\delta_{A}(\beta^{i}\alpha^{j}) \\
    &=& \frac{q-1}{\mu}\omega^{\frac{q-1}{\mu}\nu i+\frac{r+1}{2}(s-1)-\frac{\nu(r+1)}{\mu}\big((s-2)i+\frac{s(s-1)}{2}\big)+k(r+1)},
\end{eqnarray*}
for some integer $k$. From conditions (ii)-(iv),
\[\eta\big(\delta_{A'}(\beta^{i}\alpha^{j})\big)=1.\]
From the proof of Theorem \ref{thm5}, $\pi_{B}(\beta^{i}\alpha^{j}) \in \mathbb{F}_{r}$ which is a square. Thus
\[\eta\big(\delta_{A'}(\beta^{i}\alpha^{j})\pi_{B}(\beta^{i}\alpha^{j})\big)=1.\]
Finally, for any $\alpha^{i}\beta^{j} \in B \setminus A'$, by \eqref{8} and \eqref{9}, there exists an integer $k$ such that
\[\alpha^{i}\beta^{j}\pi_{A}(\alpha^{i}\beta^{j})\delta_{B}(\alpha^{i}\beta^{j})=\frac{q-1}{\nu}\omega^{\frac{r+1}{2}s-\frac{\nu(r+1)}{\mu}\big(sj+\frac{s(s-1)}{2}\big)+k(r+1)+\frac{q-1}{\nu}\mu i}\Omega,\]
where $\Omega=\prod_{\ell \neq i}(\omega^{\frac{q-1}{\nu}\mu j}-\omega^{\frac{q-1}{\nu}\mu\ell}) \in \mathbb{F}_{r}$. By conditions (i)-(iii), $\alpha^{i}\beta^{j}\pi_{A}(\alpha^{i}\beta^{j})\delta_{B}(\alpha^{i}\beta^{j})$ is a square in $\mathbb{F}_{q}$. Thus
\begin{eqnarray*}
  \eta\big(\pi_{A'}(\alpha^{i}\beta^{j})\delta_{B}(\alpha^{i}\beta^{j})\big) &=& \eta\big(\alpha^{i}\beta^{j}\pi_{A}(\alpha^{i}\beta^{j})\delta_{B}(\alpha^{i}\beta^{j})\big) \\
    &=& 1.
\end{eqnarray*}
In summary, by Lemma \ref{lem5}, for any $e \in A\bigtriangleup B$,
\[\eta\big(\delta_{A'\bigtriangleup B}(e)\big)=\eta\big(-\delta_{A'\bigtriangleup B}(e)\big)=1.\]
The theorem then follows from Lemma \ref{lem3}.
\end{proof}

Theorems \ref{thm5}, \ref{thm6} and \ref{thm7} provide general and powerful constructions of MDS Euclidean self-dual codes. By choosing different pairs $(\mu, \nu)$, we can obtain several explicit constructions of MDS Euclidean self-dual codes with flexible parameters. We list some pairs $(\mu, \nu)$ with specific values in the following theorems.

 \begin{theorem}\label{thm8}
Suppose $q=r^{2}$ and  $r \equiv 3 (\textnormal{mod } 4)$. Let  $0 \leq s \leq \frac{r+1}{2}$ and $0 \leq t \leq \frac{r-1}{2}$. Put $n=s(r-1)+t(r+1)-2st$, then there exist two $q$-ary MDS Euclidean self-dual codes of lengths $n$ and $n+2$, respectively.
\end{theorem}

\begin{proof}
Let $\mu=r+1$ and $\nu=r-1$. Then $\frac{\mu}{\gcd(\mu, \nu)}=\frac{r+1}{2}$, $\frac{\nu}{\gcd(\mu, \nu)}= \frac{r-1}{2}$ and  $s\frac{q-1}{\mu}+t\frac{q-1}{\nu}-\frac{2(q-1)\gcd(\mu,\nu)}{\mu\nu}st=s(r-1)+t(r+1)-2st=n$.  It can be verified that the conditions in Theorems \ref{thm5} and \ref{thm7} hold. Thus the theorem follows.
\end{proof}

\begin{example}
In Theorem \ref{thm8}, let $r=19$, $q=r^{2}=361$, $s=3$ and $t=5$. Then  $n=s(r-1)+t(r+1)-2st=124$. So we can obtain two MDS Euclidean self-dual codes of length $124$ and $126$ over $\mathbb{F}_{361}$, respectively. The parameters of the two codes are new in the sense that they have not been obtained in the literature (see Table I).
\end{example}

\begin{theorem}\label{thm9}
Suppose $q=r^{2}$ and $r \equiv 3 (\textnormal{mod } 4)$. Let  $0 \leq s \leq r+1$ and $0 \leq t \leq \frac{r-1}{2}$. Put $n=s(r-1)+2t(r+1)-4st$.
 \begin{description}
   \item[(i)] If $s$ is odd, then there exists a $q$-ary MDS Euclidean self-dual code of length $n$;
   \item[(ii)] if $s \equiv 0 (\textnormal{mod } 4)$, then there exists a $q$-ary MDS Euclidean self-dual code of length $n+2$.
 \end{description}
\end{theorem}

\begin{proof}
Let $\mu=r+1$ and $\nu=\frac{r-1}{2}$. Then $\frac{\mu}{\gcd(\mu, \nu)}=r+1$, $\frac{\nu}{\gcd(\mu, \nu)}= \frac{r-1}{2}$ and  $s\frac{q-1}{\mu}+t\frac{q-1}{\nu}-\frac{2(q-1)\gcd(\mu,\nu)}{\mu\nu}st=s(r-1)+2t(r+1)-4st=n$.
\begin{description}
  \item[(i)] If $s$ is odd, it can be verified that the conditions in Theorem \ref{thm5} hold;
  \item[(ii)] If $s \equiv 0 (\textnormal{mod } 4)$, it can be verified that the conditions in Theorem \ref{thm7} hold. Thus the theorem follows.
\end{description}
\end{proof}

\begin{example}
In Theorem \ref{thm9}, let $r=11$ and $q=r^{2}=121$.
\begin{description}
  \item[(i)] Let $s=t=3$. Then  $n=s(r-1)+2t(r+1)-4st=66$. So we can obtain an MDS Euclidean self-dual code of length $66$ over $\mathbb{F}_{121}$ from Theorem \ref{thm9} (i).
  \item[(ii)] Let $s=4$ and $t=2$. Then $n=s(r-1)+2t(r+1)-4st=56$. So we can obtain an MDS Euclidean self-dual code of length $58$ over $\mathbb{F}_{121}$ from Theorem \ref{thm9} (ii).
\end{description}
The parameters of the two code are new in the sense that they have not been obtained in the literature (see Table I).
\end{example}

\begin{theorem}\label{thm10}
Suppose $q=r^{2}$ and $r \equiv 3 (\textnormal{mod } 4)$. Let  $0 \leq s \leq \frac{r+1}{4}$ and $0 \leq t \leq \frac{r-1}{2}$. Put $n=2s(r-1)+t(r+1)-8st$. Then there exist two $q$-ary MDS Euclidean self-dual codes of lengths $n$ and $n+2$, respectively.
\end{theorem}

\begin{proof}
Let $\mu=\frac{r+1}{2}$ and $\nu=r-1$. Then $\frac{\mu}{\gcd(\mu, \nu)}=\frac{r+1}{4}$, $\frac{\nu}{\gcd(\mu, \nu)}= \frac{r-1}{2}$ and  $s\frac{q-1}{\mu}+t\frac{q-1}{\nu}-\frac{2(q-1)\gcd(\mu,\nu)}{\mu\nu}st=2s(r-1)+t(r+1)-8st=n$.  It can be verified that the conditions in Theorems \ref{thm5} and \ref{thm7} hold. Thus the theorem follows.
\end{proof}

\begin{example}
In Theorem \ref{thm10}, let $r=27$, $q=r^{2}=729$, $s=2$ and $t=1$. Then $n=2s(r-1)+t(r+1)-8st=116$.
So we can obtain two MDS Euclidean self-dual codes of lengths $116$ and $118$ over $\mathbb{F}_{729}$, respectively. The parameters of the two codes are new in the sense that they have not been obtained in the literature (see Table I).
\end{example}

\begin{theorem}\label{thm11}
Suppose $q=r^{2}$ and $r=p^{m}$ with $r \equiv 3 (\textnormal{mod } 4)$. Let $0 \leq s \leq r+1$ and $0 \leq t \leq \frac{r-1}{2}$. Put $n=s\frac{r-1}{2}+t(r+1)-2st$.
\begin{description}
  \item[(i)] If $s$ is even, then there exists a $q$-ary MDS Euclidean self-dual code of length $n$;
  \item[(ii)] If $s \equiv 1 (\textnormal{mod } 4)$, then there exists a $q$-ary MDS Euclidean self-dual code of length $n+1$;
  \item[(iii)] If $s \equiv 0 (\textnormal{mod } 4)$, then there exists a $q$-ary MDS Euclidean self-dual code of length $n+2$.
\end{description}
\end{theorem}

\begin{proof}
Let $\mu=2(r+1)$ and $\nu=r-1$. Then $\frac{\mu}{\gcd(\mu, \nu)}=r+1$, $\frac{\nu}{\gcd(\mu, \nu)}= \frac{r-1}{2}$ and  $s\frac{q-1}{\mu}+t\frac{q-1}{\nu}-\frac{2(q-1)\gcd(\mu,\nu)}{\mu\nu}st=s\frac{r-1}{2}+t(r+1)-2st=n$.
\begin{description}
  \item[(i)]  If $s$ is even, then $n$ is even and it can be verified that the conditions in Theorem \ref{thm5} hold;
  \item[(ii)] If $s \equiv 1 (\textnormal{mod } 4)$, then $n$ is odd and it can be verified that the conditions in Theorem \ref{thm6} hold;
  \item[(iii)] If $s \equiv 0 (\textnormal{mod } 4)$, then $n$ is even and it can be verified that the conditions in Theorem \ref{thm7} hold.
\end{description}
Thus the theorem follows.
\end{proof}

\begin{example}
In Theorem \ref{thm11}, let $r=19$ and $q=r^{2}=361$.
\begin{description}
  \item[(i)] Let $s=8$ and $t=4$. Then  $n=s\frac{r-1}{2}+t(r+1)-2st=88$. So we can obtain an MDS Euclidean self-dual code of length $88$ over $\mathbb{F}_{361}$ from Theorem \ref{thm11} (i).
  \item[(ii)] Let $s=13$ and $t=1$. Then $n=s\frac{r-1}{2}+t(r+1)-2st=111$. So we can obtain an MDS Euclidean self-dual code of length $112$ over $\mathbb{F}_{361}$ from Theorem \ref{thm11} (ii).
  \item[(iii)] Let $s=12$ and $t=1$. Then $n=s\frac{r-1}{2}+t(r+1)-2st=104$. So we can obtain an MDS Euclidean self-dual code of length $106$ over $\mathbb{F}_{361}$ from Theorem \ref{thm11} (iii).
\end{description}
The parameters of the above three codes are new in the sense that they have not been obtained in the literature (see Table I).
\end{example}

\begin{theorem}\label{thm12}
Suppose $q=r^{2}$ and $r=p^{m}$ with $r \equiv 3 (\textnormal{mod } 4)$. Let $0 \leq s \leq \frac{r+1}{4}$ and $0 \leq t \leq \frac{r-1}{2}$. Put $n=s(r-1)+t\frac{r+1}{2}-4st$.
Then there exist two $q$-ary MDS Euclidean self-dual codes of lengths $n$ and $n+2$, respectively.
\end{theorem}

\begin{proof}
Let $\mu=r+1$ and $\nu=2(r-1)$. Then $\frac{\mu}{\gcd(\mu, \nu)}=\frac{r+1}{4}$, $\frac{\nu}{\gcd(\mu, \nu)}= \frac{r-1}{2}$ and  $s\frac{q-1}{\mu}+t\frac{q-1}{\nu}-\frac{2(q-1)\gcd(\mu,\nu)}{\mu\nu}st=s(r-1)+t\frac{r+1}{2}-4st=n$.  It can be verified that the conditions in Theorems \ref{thm5} and \ref{thm7} hold. Thus the theorem follows.
\end{proof}

\begin{example}
In Theorem \ref{thm12}, let $r=31$ and $q=r^{2}=961$.
\begin{description}
  \item[(i)] Let $s=7$ and $t=1$. Then  $n=s(r-1)+t\frac{r+1}{2}-4st=198$. So we can obtain an MDS Euclidean self-dual code of length $198$ over $\mathbb{F}_{961}$ from Theorem \ref{thm12}.
  \item[(ii)] Let $s=2$ and $t=14$. Then $n=s(r-1)+t\frac{r+1}{2}-4st=172$. So we can obtain an MDS Euclidean self-dual code of length $174$ over $\mathbb{F}_{961}$ from Theorem \ref{thm12}.
\end{description}
The parameters of the above two codes are new in the sense that they have not been obtained in the literature (see Table I).
\end{example}

\begin{theorem}\label{thm13}
Suppose $q=r^{2}$ and $r=p^{m}$ with $r \equiv 3 (\textnormal{mod } 4)$. Let  $0 \leq s \leq \frac{r+1}{2}$ and $0 \leq t \leq \frac{r-1}{2}$. Put $n=2s(r-1)+2t(r+1)-8st$. Then there exists a $q$-ary MDS Euclidean self-dual code of length $n+2$.
\end{theorem}

\begin{proof}
Let $\mu=\frac{r+1}{2}$ and $\nu=\frac{r-1}{2}$. Then $\frac{\mu}{\gcd(\mu, \nu)}=\frac{r+1}{2}$, $\frac{\nu}{\gcd(\mu, \nu)}= \frac{r-1}{2}$ and  $s\frac{q-1}{\mu}+t\frac{q-1}{\nu}-\frac{2(q-1)\gcd(\mu,\nu)}{\mu\nu}st=2s(r-1)+2t(r+1)-8st=n$.
It can be verified that the conditions in Theorem \ref{thm7} hold. Thus the theorem follows.
\end{proof}

\begin{example}
In Theorem \ref{thm13}, let $r=31$, $q=r^{2}=961$, $s=9$ and $t=2$. Then  $n=2s(r-1)+2t(r+1)-8st=524$. So we can obtain an MDS Euclidean self-dual code of length $526$ over $\mathbb{F}_{961}$ from Theorem \ref{thm13}.
The parameters of this code are new in the sense that they have not been obtained in the literature (see Table I).
\end{example}

\begin{theorem}\label{thm14}
Suppose $q=r^{2}$ and $r=p^{m}$ with $r \equiv 3 (\textnormal{mod } 4)$. Let  $0 \leq s \leq \frac{r+1}{2}$ and $0 \leq t \leq \frac{r-1}{2}$. Put $n=s\frac{r-1}{2}+t\frac{r+1}{2}-2st$.
\begin{description}
  \item[(i)] If $s$ is even, then there exists a $q$-ary MDS Euclidean self-dual code of length $n$;
  \item[(ii)] If $s$ is odd, then there exists a $q$-ary MDS Euclidean self-dual code of length $n+1$;
  \item[(iii)] If $s$ is even, then there exists a $q$-ary MDS Euclidean self-dual code of length $n+2$.
\end{description}
\end{theorem}

\begin{proof}
Let $\mu=2(r+1)$ and $\nu=2(r-1)$. Then $\frac{\mu}{\gcd(\mu, \nu)}=\frac{r+1}{2}$, $\frac{\nu}{\gcd(\mu, \nu)}= \frac{r-1}{2}$ and  $s\frac{q-1}{\mu}+t\frac{q-1}{\nu}-\frac{2(q-1)\gcd(\mu,\nu)}{\mu\nu}st=s\frac{r-1}{2}+t\frac{r+1}{2}-2st=n$.
\begin{description}
  \item[(i)] If $s$ is even, then $n$ is even and it can be verified that the conditions in Theorem \ref{thm5} hold;
  \item[(ii)] If $s$ is odd, then $n$ is odd and it can be verified that the conditions in Theorem \ref{thm6} hold;
  \item[(iii)] If $s$ is even, then $n$ is even and it can be verified that the conditions in Theorem \ref{thm7} hold.
\end{description}
Thus the theorem follows.
\end{proof}

\begin{example}
In Theorem \ref{thm14}, let $r=23$ and $q=r^{2}=529$.
\begin{description}
  \item[(i)] Let $s=10$ and $t=1$. Then  $n=s\frac{r-1}{2}+t\frac{r+1}{2}-2st=102$. So we can obtain an MDS Euclidean self-dual code of length $102$ over $\mathbb{F}_{529}$ from Theorem \ref{thm14} (i).
  \item[(ii)] Let $s=9$ and $t=4$. Then $n=s\frac{r-1}{2}+t\frac{r+1}{2}-2st=75$. So we can obtain an MDS Euclidean self-dual code of length $76$ over $\mathbb{F}_{529}$ from Theorem \ref{thm14} (ii).
  \item[(iii)] Let $s=10$ and $t=1$. Then  $n=s\frac{r-1}{2}+t\frac{r+1}{2}-2st=102$. So we can obtain an MDS Euclidean self-dual code of length $104$ over $\mathbb{F}_{529}$ from Theorem \ref{thm14} (iii).
\end{description}
The parameters of the above three codes are new in the sense that they have not been obtained in the literature (see Table I).
\end{example}

\begin{theorem}\label{thm15}
Suppose $q=r^{2}$ and $r=p^{m}$ with $r \equiv 3 (\textnormal{mod } 4)$. Suppose $r+1=2^{a}\cdot b$, where $b$ is odd. Let $0 \leq s \leq b$ and $0 \leq t \leq r-1$. Put $n=s(r-1)+tb-2st$.
\begin{description}
  \item[(i)] If $t$ is even, then there exists a $q$-ary MDS Euclidean self-dual code of length $n$;
  \item[(ii)] If $t$ is odd, then there exists a $q$-ary MDS Euclidean self-dual code of length $n+1$;
  \item[(iii)] If $t$ is even, then there exists a $q$-ary MDS Euclidean self-dual code of length $n+2$.
\end{description}
\end{theorem}

\begin{proof}
Let $\mu=r+1$ and $\nu=2^{a}(r-1)$. Then $\frac{\mu}{\gcd(\mu, \nu)}=\frac{r+1}{2^{a}}=b$, $\frac{\nu}{\gcd(\mu, \nu)}= r-1$ and  $s\frac{q-1}{\mu}+t\frac{q-1}{\nu}-\frac{2(q-1)\gcd(\mu,\nu)}{\mu\nu}st=s(r-1)+tb-2st=n$.
\begin{description}
  \item[(i)] If $t$ is even, then $n$ is even and it can be verified that the conditions in Theorem \ref{thm5} hold;
  \item[(ii)] If $t$ is odd, then $n$ is odd and it can be verified that the conditions in Theorem \ref{thm6} hold;
  \item[(iii)] If $t$ is even, then $n$ is even and it can be verified that the conditions in Theorem \ref{thm7} hold.
\end{description}
Thus the theorem follows.
\end{proof}

\begin{example}
In Theorem \ref{thm15}, let $r=27$ and $q=r^{2}=729$.
\begin{description}
  \item[(i)] Let $s=2$ and $t=8$. Then  $n=s(r-1)+t\frac{r+1}{2^{a}}-2st=76$. So we can obtain an MDS Euclidean self-dual code of length $76$ over $\mathbb{F}_{729}$ from Theorem \ref{thm15} (i).
  \item[(ii)] Let $s=4$ and $t=11$. Then  $n=s(r-1)+t\frac{r+1}{2^{a}}-2st=93$. So we can obtain an MDS Euclidean self-dual code of length $94$ over $\mathbb{F}_{729}$ from Theorem \ref{thm15} (ii).
  \item[(iii)]Let $s=6$ and $t=16$. Then  $n=s(r-1)+t\frac{r+1}{2^{a}}-2st=66$. So we can obtain an MDS Euclidean self-dual code of length $68$ over $\mathbb{F}_{729}$ from Theorem \ref{thm15} (iii).
\end{description}
The parameters of the above three codes are new in the sense that they have not been obtained in the literature (see Table 1).
\end{example}

\begin{theorem}\label{thm16}
Suppose $q=r^{2}$ and $r=p^{m}$ with $r \equiv 3 (\textnormal{mod } 4)$. Suppose $r+1=2^{a}\cdot b$, where $b$ is odd. Let $0 \leq s \leq b$ and $0 \leq t \leq r-1$. Put $n=s(r-1)2^{a}+t(r+1)-2st$.
Then there exist two $q$-ary MDS Euclidean self-dual codes of lengths $n$ and $n+2$, respectively.
\end{theorem}

\begin{proof}
Let $\mu=b$ and $\nu=r-1$. Then $\frac{\mu}{\gcd(\mu, \nu)}=b$, $\frac{\nu}{\gcd(\mu, \nu)}= r-1$ and  $s\frac{q-1}{\mu}+t\frac{q-1}{\nu}-\frac{2(q-1)\gcd(\mu,\nu)}{\mu\nu}st=s(r-1)2^{a}+t(r+1)-2st=n$.
Then it can be verified that the conditions in Theorems \ref{thm5} and \ref{thm7} hold. Thus the theorem follows.
\end{proof}

\begin{example}
In Theorem \ref{thm16}, let $r=23$, $q=r^{2}=529$, $s=5$ and $t=4$. Then  $n=s(r-1)+t\frac{r+1}{2^{a}}-2st=376$. So we can obtain two MDS Euclidean self-dual codes of lengths $376$ and $378$ over $\mathbb{F}_{529}$ from Theorem \ref{thm16},respectively.
The parameters of these two codes are new in the sense that they have not been obtained in the literature (see Table I).
\end{example}
\section{Conclusion}
In this paper, we give a further research on construction of MDS Euclidean self-dual codes. The main tool of our results is the (extended) GRS codes and the key point of our construction is to choose suitable evaluation sets such that the corresponding (extended) GRS codes are Euclidean self-dual codes. By utilizing some multiplicative subgroups, additive subgroups and the trace function of finite fields, we provide several new families of MDS Euclidean self-dual codes which are not covered in the literature. In particular, we provide a useful lemma (see Lemma \ref{lem5}) to ensure that the symmetric difference of two intersecting subsets satisfies
 the conditions in Lemmas \ref{lem2} or \ref{lem3}, which can produce new MDS Euclidean self-dual codes. Based on Lemma \ref{lem5}, by using two multiplicative subgroups and their cosets, we present three generic and powerful constructions of MDS Euclidean self-dual codes with flexible parameters. Then several new classes MDS Euclidean self-dual codes are explicitly constructed. It will be interesting to employ other subsets of finite fields with special structures satisfying the conditions of Lemma \ref{lem5}. And generalizing the results of Lemma \ref{lem5} for more than two subsets and find more new constructions of MDS Euclidean self-dual codes will be one of the future research direction.

\section*{Acknowledge}
The research of W. Fang and  S.-T. Xia is supported in part by the National Key Research and Development Program of China under Grant 2018YFB1800204, the National Natural Science Foundation of China under Grant 61771273, Guangdong Basic and Applied Basic Research Foundation under Grant 2019A1515110904, the R$\&$D Program of Shenzhen under Grant JCYJ20180508152204044, and the project ``PCL Future Greater-Bay Area Network Facilities for Large-scale Experiments and Applications (LZC0019)". The research of F.-W. Fu is supported in part by the National Key Research and Development Program of China under Grant 2018YFA0704703, the National Natural Science Foundation of China under Grant Nos. 61971243, 61571243, and the Nankai Zhide Foundation.

\end{document}